\documentclass[10pt,twocolumn]{IEEEtran}

\usepackage{amsmath,amssymb,amsfonts}
\usepackage{bm}
\usepackage{cite}
\usepackage{graphicx}
\usepackage{algorithm}
\usepackage{algorithmic}
\usepackage{url}
\usepackage{amsthm}
\usepackage{hyperref}
\usepackage[caption=false,font=footnotesize]{subfig}

\newtheorem{proposition}{Proposition}
\newtheorem{corollary}{Corollary}

\title{\huge Continuous Fluid Antenna Sampling for Channel Estimation in Cell-Free Massive MIMO}

	\author{Masoud~Kaveh,~\IEEEmembership{Member},~\textit{IEEE},~Farshad~Rostami~Ghadi,~\IEEEmembership{Member},~\textit{IEEE}, Francisco Hernando-Gallego, ~Diego~Mart\'in, Riku Jäntti,~\IEEEmembership{Senior Member},~\textit{IEEE}, and
		Kai-Kit~Wong,~\IEEEmembership{Fellow},~\textit{IEEE}  }

\makeatletter
\def\blfootnote{\xdef\@thefnmark{}\@footnotetext}
\makeatother
\begin{document}
	\maketitle
\blfootnote{The work of M. Kaveh and R. Jäntti has received funding from the SNS JU under the EU’s Horizon Europe Research and Innovation Programme under Grant Agreement No. 101192113 (Ambient-6G). The work of F. Rostami Ghadi is supported by the European Union's Horizon 2022 Research and Innovation Programme under Marie Skłodowska-Curie Grant No. 101107993. The work of  K. K. Wong is supported by the Engineering and Physical Sciences Research Council (EPSRC) under Grant EP/W026813/1.}

\blfootnote{\noindent M. Kaveh and R. Jäntti are with the Department of Information and Communication Engineering, Aalto University, Espoo, Finland. (e-mail: $\rm masoud.kaveh@aalto.fi, riku.jantti@aalto.fi$).}
	\blfootnote{\noindent F. Rostami Ghadi is with the Department of Signal Theory, Networking and Communications, 
    University of Granada, 18071, Granada, Spain. (e-mail: $\rm f.rostami@ugr.es$).}
\blfootnote{\noindent F. H. Gallego and D. Mart\'in are with the Department of Applied Mathematics and Computer Science, respectively, 
Escuela de Ingenier\'ia Inform\'atica de Segovia, 
Universidad de Valladolid, Segovia, Spain (e-mail:$\rm \{fhernando, diego.martin.andres\}@uva.es$)}
\blfootnote{\noindent K. K. Wong is affiliated with the Department of Electronic and Electrical Engineering, University College London, Torrington Place, WC1E 7JE, United Kingdom and also affiliated with Yonsei Frontier Lab, Yonsei University, Seoul, Republic of Korea (e-mail: $\rm kai-kit.wong@ucl.ac.uk$).}
	\begin{abstract}
		In this letter, we develop a continuous fluid antenna (FA) framework for uplink channel estimation in cell-free massive multiple-input and multiple-output (CF-mMIMO) systems. By modeling the wireless channel as a spatially correlated Gaussian random field, channel estimation is formulated as a Gaussian process (GP) regression problem with motion-constrained spatial sampling. Closed-form expressions for the linear minimum mean squared error (LMMSE) estimator and the corresponding estimation error are derived. A fundamental comparison with discrete port-based architectures is established under identical position constraints, showing that continuous FA sampling achieves equal or lower estimation error for any finite pilot budget, with strict improvement for non-degenerate spatial correlation models. Numerical results validate the analysis and show the performance gains of continuous FA sampling over discrete baselines.
	\end{abstract}

	\vspace{-5mm}
	\section{Introduction}
	
	\IEEEPARstart{C}{ell}-free massive multiple-input multiple-output (CF-mMIMO) systems employ a large number of geographically distributed access points (APs) that jointly serve users without cell boundaries\cite{el2022cell}. By coherently combining signals from multiple APs, CF-mMIMO provides significant improvements in spectral efficiency, macro-diversity, and user fairness compared to conventional cellular systems \cite{ngo2017cell}. These benefits, however, rely critically on the availability of accurate channel state information (CSI), making uplink channel estimation a central challenge in CF-mMIMO deployments \cite{ammar2022user}.
	
	Fluid antennas systems (FAS) have emerged as a promising technology to exploit spatial channel variations without increasing hardware complexity\cite{wong2021fluid,Ghadi2024:FAS_BC}. By allowing the antenna radiation point to change the position within a confined region, FAS enable the wireless channel to be sampled across space using a single radio frequency (RF) chain. Since wireless channels exhibit significant spatial variations over distances on the order of the carrier wavelength, even small antenna displacements can result in different channel realizations \cite{new2025tut}.
	
	Most existing FAS-assisted studies discretize the fluid motion into a finite set of antenna ports and optimize port selection during data transmission or uplink training \cite{xu2024channel,new2025channel}. While such discrete port-based architectures provide performance gains over fixed-position antennas (FPA), they suffer from two fundamental limitations. First, they approximate the inherently continuous motion of the fluid antenna using a finite set of candidate positions, leading to information loss. Second, they require frequent port switching, which may introduce latency, energy consumption, and hardware constraints, particularly when switching occurs at symbol-level time scales.
	
	Recent works have proposed continuous fluid antenna (FA) models to analyze fading statistics, outage probability, and level-crossing behavior in single-link scenarios \cite{psom2023cont,smith2025dim}. However, a unified and rigorous treatment of continuous FA motion in CF-mMIMO systems, particularly from the perspective of uplink channel estimation, remains largely unexplored. In particular, it is unclear to what extent discrete port selection fundamentally limits channel estimation accuracy compared to continuous FA sampling under realistic position constraints.

    \begin{figure}[t]
    \centering
    \includegraphics[width=0.849\linewidth]{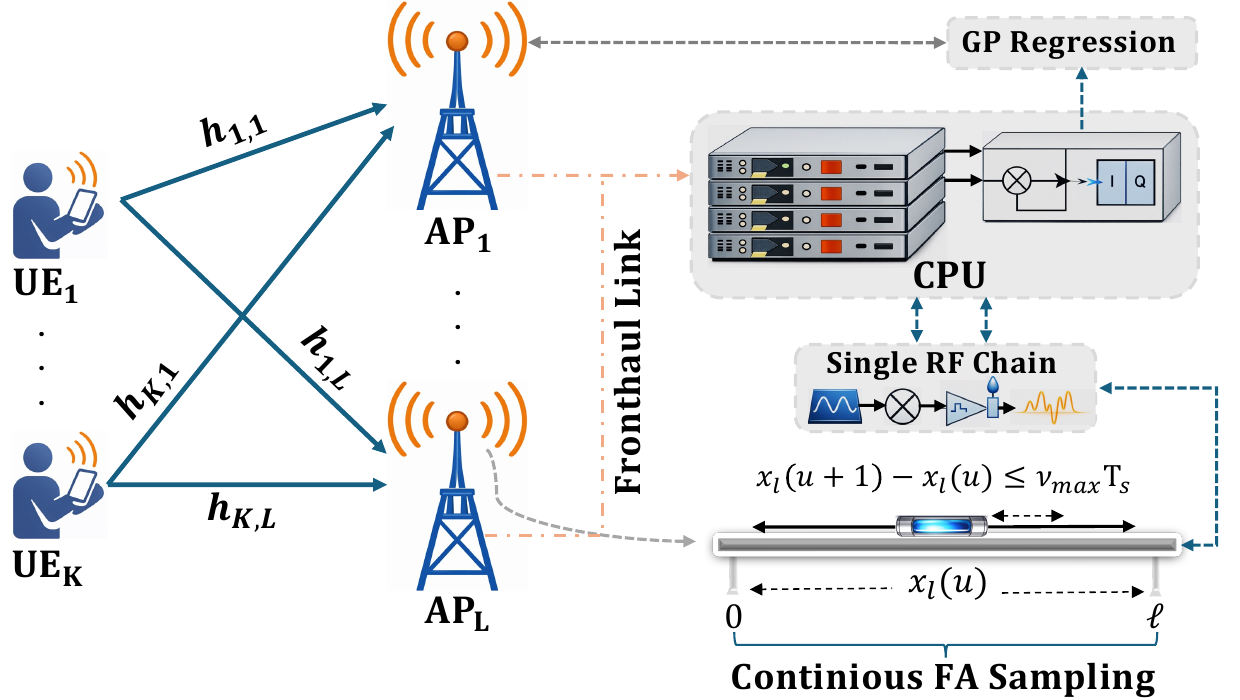}
    \vspace{-8 pt}
    \caption{Uplink CF-mMIMO with continuous fluid antenna sampling.}
    \label{fig:sysmodel}
        \vspace{-16 pt}
\end{figure}

	Addressing the limitations of discrete port-based FAS, this paper develops a continuous FA framework for uplink channel estimation in CF-mMIMO systems, as shown in Fig. \ref{fig:sysmodel}. Each AP is equipped with a single FA whose radiating element changes the position continuously during the training phase, enabling spatial channel sampling without discretization. By modeling the wireless channel as a spatially correlated complex Gaussian random field, uplink channel estimation is formulated as a Gaussian process (GP) regression problem with motion-constrained sampling. We derive closed-form expressions for the linear minimum mean-squared error (LMMSE) estimator and the corresponding estimation error, explicitly revealing the role of antenna motion and spatial correlation. A fundamental performance comparison with discrete port-based architectures is then established. It is shown that, for any finite pilot budget, continuous FA sampling achieves an estimation error no larger than that of any discrete scheme with a finite number of ports, with strict improvement for non-degenerate spatial correlation models. Numerical results validate the analysis and illustrate the resulting gains under practical system parameters.

\vspace{-3mm}
	\section{System Model}\label{sec:system}
	\subsection{Network Model}
	As shown in Fig. \ref{fig:sysmodel}, we consider the uplink of a CF-mMIMO system composed of $L$ geographically distributed APs jointly serving $K$ single-antenna user equipments (UEs) on the same time-frequency resource. We assume the APs are connected to a central processing unit (CPU) through fronthaul links and cooperate in channel estimation and data detection. Also, time-division duplexing (TDD) operation is assumed, such that uplink and downlink channels are reciprocal. Each AP is equipped with a single FA, whose radiating element can continuously change the position along a one-dimensional segment of finite length $\ell$, normalized by the carrier wavelength. The FA is connected to a single RF chain, so that at any given time instant only one antenna position is active at each AP. The FA motion is controlled locally at each AP and is assumed to be synchronized with the uplink pilot transmission phase.
	
	Following the user-centric CF-mMIMO paradigm \cite{ammar2022user}, each user $k \in \{1,\dots,K\}$ is served only by a subset of APs denoted by $\mathcal{L}_k \subseteq \{1,\dots,L\}$, which is determined based on large-scale fading coefficients. Conversely, each AP $l$ serves a subset of users denoted by $\mathcal{K}_l$. This association is assumed to remain fixed over the considered coherence block. Each UE transmits pilot symbols for channel estimation followed by uplink data symbols. Additionally, perfect synchronization among UEs and APs is assumed, and all signals are narrowband such that frequency-flat fading applies.\vspace{-5mm}
	\subsection{Continuous FA Channel Model}
	
	Let $x_l \in [0,\ell]$ denote the instantaneous position of the fluid antenna at AP $l$, measured along the antenna’s movement axis and normalized by the carrier wavelength. The uplink channel coefficient between UE $k$ and AP $l$ when the FA is located at position $x_l$ is denoted by $h_{k,l}(x_l) \in \mathbb{C}$.
	
	We model the small-scale fading component of the uplink channel as a spatially continuous complex Gaussian random field indexed by the antenna position, in contrast to discrete parametric sparse-array models. Specifically, for each UE-AP pair $(k,l)$, the channel $h_{k,l}(x)$ is modeled as a zero-mean, circularly symmetric, spatially correlated complex Gaussian random field\footnote{The channel process $h_{k,l}(x)$ is assumed to be a proper complex Gaussian random field, i.e.,
		$\mathbb{E}[h_{k,l}(x) h_{k,l}(x')] = 0$ for all $x,x'$. This assumption ensures circular symmetry and guarantees that standard LMMSE and GP regression results apply without modification.} satisfying
	$
		\mathbb{E}\!\left[h_{k,l}(x)\right] = 0,
	$
	and
	\begin{equation}
		\mathbb{E}\!\left[h_{k,l}(x) h_{k,l}^*(x')\right]
		= \beta_{k,l} \, \kappa(x,x'),
		\label{eq:spatial_covariance}
	\end{equation}
	where $\beta_{k,l} > 0$ is the large-scale fading coefficient accounting for path-loss and shadowing between UE $k$ and AP $l$, and $\kappa(x,x')$ is a normalized spatial correlation kernel that depends only on the distance between antenna positions. The large-scale fading coefficients $\{\beta_{k,l}\}$ are assumed to be constant over many coherence blocks and known at the CPU and APs, as commonly assumed in CF-mMIMO systems \cite{ngo2017cell}. In contrast, the small-scale fading realizations vary from block to block.
	
	Throughout this work, unless otherwise stated, we adopt a Jakes's model to characterize the spatial correlation \cite{new2025tut}, i.e., 
	\begin{align}
		\kappa(x,x') = J_0\!\left(2\pi |x - x'|\right),
		\label{eq:jakes_kernel}
	\end{align}
	where $J_0(\cdot)$ is the zero-order Bessel function of the first kind. \vspace{-3mm} 
	
	
	
	
	\vspace{-3mm}
	\section{Uplink Training and Continuous FA Sampling}\label{sec:uplink}
	
	\subsection{Pilot Transmission Model}
	
	We consider an uplink training phase of $\tau_p$ channel uses within each coherence block. Each UE $k \in \{1,\dots,K\}$ transmits a deterministic pilot sequence
	$\boldsymbol{\phi}_k = [\phi_k(1), \dots, \phi_k(\tau_p)]^T \in \mathbb{C}^{\tau_p}$
	with transmit power $\eta_{p,k}$. The pilot sequences are assumed to be mutually orthogonal, i.e.,
	\begin{align}
		\boldsymbol{\phi}_k^H \boldsymbol{\phi}_{k'} =
		\begin{cases}
			\tau_p, & k = k', \\
			0, & k \neq k',
		\end{cases}
	\end{align}
	which requires $\tau_p \ge K$.
	During the training phase, the FA at AP $l$ continuously change the position along its admissible spatial domain. Let $x_l(u) \in [0,\ell]$ denote the FA position at AP $l$ during the $u$-th pilot symbol, where $u \in \{1,\dots,\tau_p\}$. The FA position is allowed to vary from symbol to symbol, enabling spatial sampling of the channel during pilot transmission.
	The received baseband signal at AP $l$ during pilot symbol $u$ is 
	\begin{align}
		y_l(u) =
		\sum_{k=1}^K
		\sqrt{\eta_{p,k}}\,
		h_{k,l}\!\left(x_l(u)\right)
		\phi_k(u)
		+ n_l(u),
		\label{eq:uplink_pilot_rx}
	\end{align}
	where $h_{k,l}(x_l(u))$ denotes the uplink channel coefficient between UE $k$ and AP $l$ evaluated at FA position $x_l(u)$, and $n_l(u) \sim \mathcal{CN}(0,\sigma^2)$ is additive white Gaussian noise.
	
	By stacking the received samples over the training interval, the received pilot signal vector at AP $l$ can be written as
	\begin{align}
		\mathbf{y}_l =
		\sum_{k=1}^K
		\sqrt{\eta_{p,k}}
		\operatorname{diag}\!\big(
		\boldsymbol{\phi}_k
		\big)
		\mathbf{h}_{k,l}
		+ \mathbf{n}_l,\, \mathbf{n}_l \sim \mathcal{CN}(\mathbf{0}, \sigma^2 \mathbf{I}),
	\end{align}
	where
	$
	\mathbf{h}_{k,l}
	=
	\big[
	h_{k,l}(x_l(1)),
	\dots,
	h_{k,l}(x_l(\tau_p))
	\big]^T
	$.
	
	To isolate the contribution of UE $k$, AP $l$ applies pilot matched filtering by correlating $\mathbf{y}_l$ with $\boldsymbol{\phi}_k^*$, yielding
	\begin{align}
		\hspace{-3mm}\mathbf{y}_{k,l}
		=
		\frac{1}{\tau_p}
		\boldsymbol{\phi}_k^H \mathbf{y}_l
		=
		\sqrt{\eta_{p,k}}\, \mathbf{h}_{k,l}
		+ \mathbf{n}_{k,l},\, \mathbf{n}_{k,l} \sim \mathcal{CN}(\mathbf{0}, \frac{\sigma^2}{\tau_p} \mathbf{I}),
		\label{eq:pilot_matched_obs}
	\end{align}
	
	\subsection{Continuous FA Sampling Interpretation}
	
	Equation \eqref{eq:pilot_matched_obs} reveals that uplink training with a fluid antenna corresponds to observing multiple samples of a spatially correlated random field at locations $\{x_l(u)\}_{u=1}^{\tau_p}$. Unlike discrete port-based FA architectures, where the antenna position is restricted to a finite set of locations, the continuous FA enables sampling at arbitrary spatial positions within $[0,\ell]$, subject only to physical position constraints. As a result, uplink training with a continuous FA can be interpreted as a spatial sampling problem of a Gaussian random field, which naturally leads to a GP regression formulation for channel estimation.
	\subsection{position constraints}
	The FA motion is constrained by physical limitations of the fluid or actuation mechanism. In particular, we impose a finite speed constraint on the FA trajectory as
	\begin{align}
		|x_l(u+1) - x_l(u)| \le v_{\max} T_s,
		\label{eq:motion_constraint}
	\end{align}
	where $T_s$ denotes the symbol duration and $v_{\max}$ is the maximum allowable velocity of the fluid antenna.
	
	Constraint \eqref{eq:motion_constraint} couples the FA positions across pilot symbols and restricts the set of admissible sampling trajectories. This constraint distinguishes continuous FA sampling from idealized models that assume instantaneous repositioning and plays a key role in determining the achievable channel estimation performance.\vspace{-0.3cm}
	
	
	
	\section{LMMSE Channel Estimation via GP Regression}
	
	In this section, we derive the LMMSE estimator for the uplink channel coefficients under continuous FA sampling. We show that the resulting estimator admits a natural interpretation as GP regression with spatially correlated observations.\vspace{-5mm}
	\subsection{Second-Order Statistics and Covariance Structure}
	Recall from Section~\ref{sec:system} that, for each UE-AP pair $(k,l)$, the channel $h_{k,l}(x)$ is modeled as a zero-mean, circularly symmetric complex Gaussian random field with spatial covariance kernel $\kappa(x,x')$. Let
	$
	\mathbf{h}_{k,l}
	=
	\big[
	h_{k,l}(x_l(1)),\,
	\dots,\,
	h_{k,l}(x_l(\tau_p))
	\big]^T
	\in \mathbb{C}^{\tau_p}
$
	denote the vector of channel samples observed at AP $l$ during the uplink training phase. Therefore, the covariance matrix of $\mathbf{h}_{k,l}$ is given by
	\begin{align}
		\mathbf{R}_{k,l}
		=
		\mathbb{E}\!\left[
		\mathbf{h}_{k,l} \mathbf{h}_{k,l}^H
		\right],
	\end{align}
	whose $(u,v)$-th entry is
	\begin{equation}
		[\mathbf{R}_{k,l}]_{u,v}
		=
		\mathbb{E}\!\left[
		h_{k,l}(x_l(u)) h_{k,l}^*(x_l(v))
		\right]
		=
		\beta_{k,l} \kappa\!\left(x_l(u), x_l(v)\right).
		\label{eq:Rkl_entries}
	\end{equation}
	
	Consider now the channel coefficient $h_{k,l}(x)$ evaluated at an arbitrary antenna position $x \in [0,\ell]$, which may or may not coincide with any of the sampled locations $\{x_l(u)\}_{u=1}^{\tau_p}$. The cross-covariance vector between $h_{k,l}(x)$ and $\mathbf{h}_{k,l}$ is defined as
	\begin{equation}
		\mathbf{r}_{k,l}(x)
		=
		\mathbb{E}\!\left[
		\mathbf{h}_{k,l} h_{k,l}^*(x)
		\right]
		\in \mathbb{C}^{\tau_p},
	\end{equation}
	with entries
	\begin{equation}
		[\mathbf{r}_{k,l}(x)]_u
		=
		\beta_{k,l} \kappa\!\left(x, x_l(u)\right),
		\quad u = 1,\dots,\tau_p.
		\label{eq:rkl_entries}
	\end{equation}
	
	Equations \eqref{eq:Rkl_entries} and \eqref{eq:rkl_entries} completely characterize the second-order statistics of the channel samples under continuous FA motion.\vspace{-6mm}
	
	\subsection{LMMSE Channel Estimator}
	
	From Section~\ref{sec:uplink}, the matched pilot observation for UE $k$ at AP $l$ is given by
	\begin{equation}
		\mathbf{y}_{k,l}
		=
		\sqrt{\eta_{p,k}}\, \mathbf{h}_{k,l}
		+
		\mathbf{n}_{k,l},
		\label{eq:yk_l_recall}
	\end{equation}
	where $\mathbf{n}_{k,l} \sim \mathcal{CN}(\mathbf{0}, \frac{\sigma^2}{\tau_p}\mathbf{I})$ and is independent of $\mathbf{h}_{k,l}$. Since $h_{k,l}(x)$ and $\mathbf{y}_{k,l}$ are jointly complex Gaussian, the LMMSE estimator of $h_{k,l}(x)$ given $\mathbf{y}_{k,l}$ coincides with the conditional mean
$
	\hat{h}_{k,l}(x)
	=
	\mathbb{E}\!\left[
	h_{k,l}(x) \,\big|\, \mathbf{y}_{k,l}
	\right]
$. Thus, using standard LMMSE estimation results, the estimator can be written as
	\begin{align}
		\hat{h}_{k,l}(x)
		=
		\mathbf{r}_{k,l}^H(x)
		\left(
		\mathbf{R}_{k,l}
		+
		\frac{\sigma^2}{\eta_{p,k}\tau_p} \mathbf{I}
		\right)^{-1}
		\mathbf{y}_{k,l}.
		\label{eq:lmmse_estimator}
	\end{align}
	
	Equation \eqref{eq:lmmse_estimator} shows that the channel estimate at position $x$ is obtained as a linear combination of the noisy channel samples collected along the FA trajectory. The weighting coefficients depend explicitly on the spatial correlation kernel and the sampling locations $\{x_l(u)\}$.\vspace{-3mm}
	
	\subsection{GP Interpretation}
	
	The estimator in \eqref{eq:lmmse_estimator} admits a natural interpretation as GP regression. Specifically, the channel $h_{k,l}(x)$ is modeled as a realization of a proper complex GP with covariance function $\beta_{k,l}\kappa(x,x')$, while the observations $\mathbf{y}_{k,l}$ correspond to noisy samples of this GP at locations $\{x_l(u)\}$. Under this interpretation, \eqref{eq:lmmse_estimator} corresponds to the posterior mean of the GP conditioned on the observed samples, and the estimation error corresponds to the posterior variance.\vspace{-3mm}
	
	
	\subsection{Estimation Error and NMSE}
	
	The mean squared error (MSE) of the LMMSE estimate at position $x$ is defined as
	\begin{align}
	\mathrm{MSE}_{k,l}(x)
	=
	\mathbb{E}\!\left[
	\big| h_{k,l}(x) - \hat{h}_{k,l}(x) \big|^2
	\right].
	\end{align}
	Using standard LMMSE theory, the MSE admits the closed-form expression
\vspace{-4mm}	\begin{align}
		\mathrm{MSE}_{k,l}(x)
		=
		\beta_{k,l}
		-
		\mathbf{r}_{k,l}^H(x)
		\left(
		\mathbf{R}_{k,l}
		+
		\frac{\sigma^2}{\eta_{p,k}\tau_p} \mathbf{I}
		\right)^{-1}
		\mathbf{r}_{k,l}(x).
		\label{eq:mse_expression}
	\end{align}
	To facilitate comparisons across different channel conditions, we define the normalized mean squared error (NMSE) as
	\begin{align}\notag
	&	\mathrm{NMSE}_{k,l}(x)
		=
		\frac{\mathrm{MSE}_{k,l}(x)}{\beta_{k,l}}
		\\
		&=
		1
		-
		\frac{
			\mathbf{r}_{k,l}^H(x)
			\left(
			\mathbf{R}_{k,l}
			+
			\frac{\sigma^2}{\eta_{p,k}\tau_p} \mathbf{I}
			\right)^{-1}
			\mathbf{r}_{k,l}(x)
		}{\beta_{k,l}}.
		\label{eq:nmse_expression}
	\end{align}
	
	
The proposed estimator provides the channel estimate at any desired antenna position within the admissible spatial domain. In practice, the channel is evaluated at the antenna position used for data transmission, typically the final position at the end of the training phase.

	\begin{figure*}[!t]
		\centering
		\subfloat[]{%
			\includegraphics[width=0.26\linewidth]{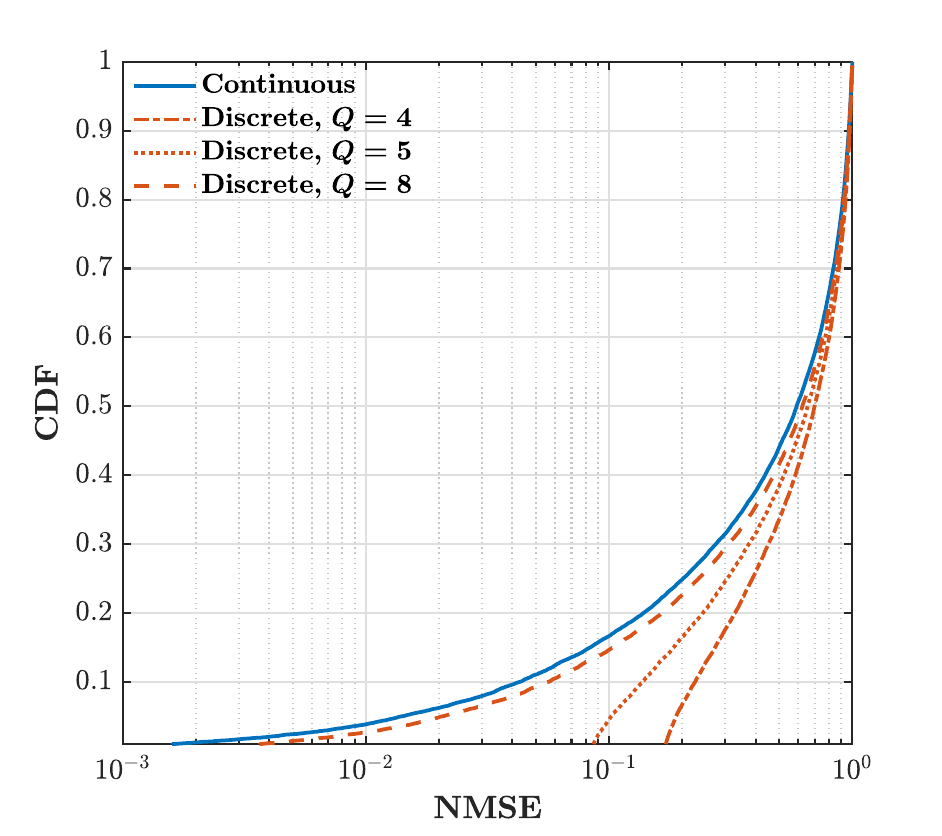}
			\label{fig:c-n}
		}\hfil\hspace{-5mm}
		\subfloat[]{%
			\includegraphics[width=0.26\linewidth]{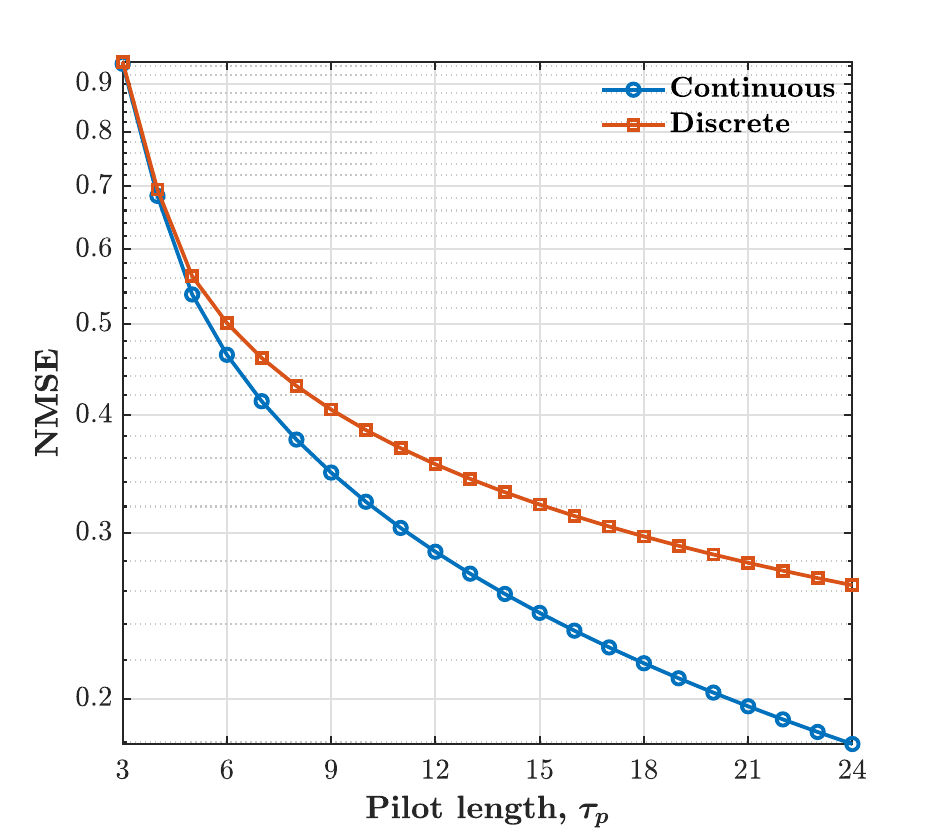}
			\label{fig:n-pi}
		}\hfil\hspace{-5mm}
		\subfloat[]{%
			\includegraphics[width=0.26\linewidth]{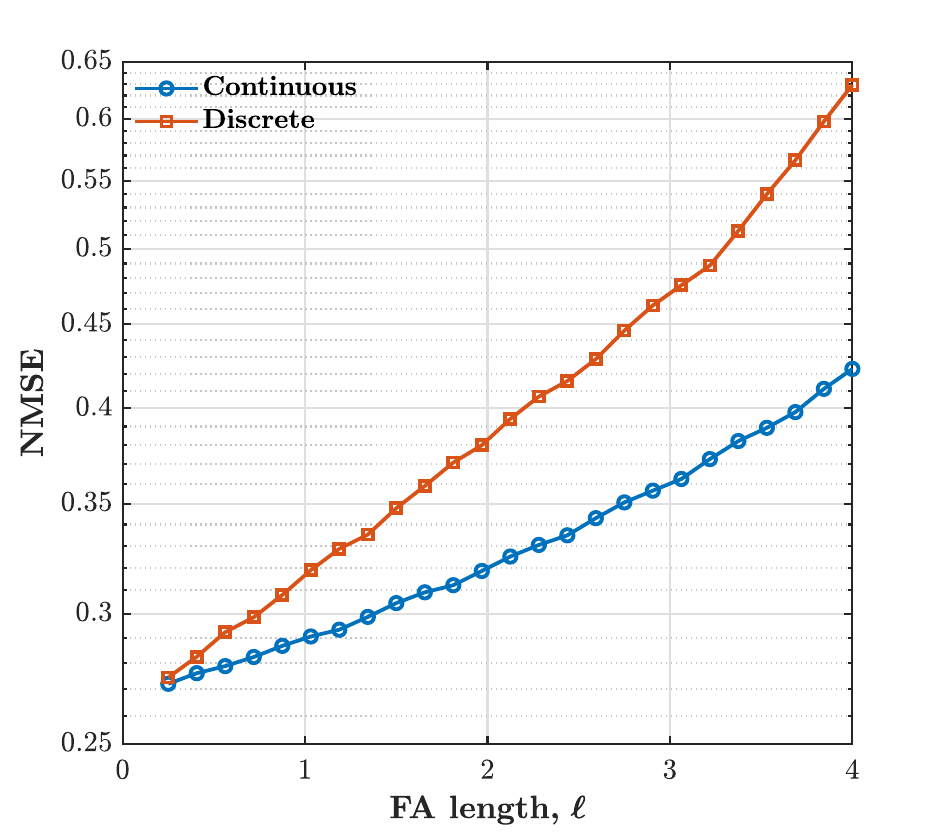}
			\label{fig:n-l}
		}\hfil\hspace{-5mm}
		\subfloat[]{%
			\includegraphics[width=0.26\linewidth]{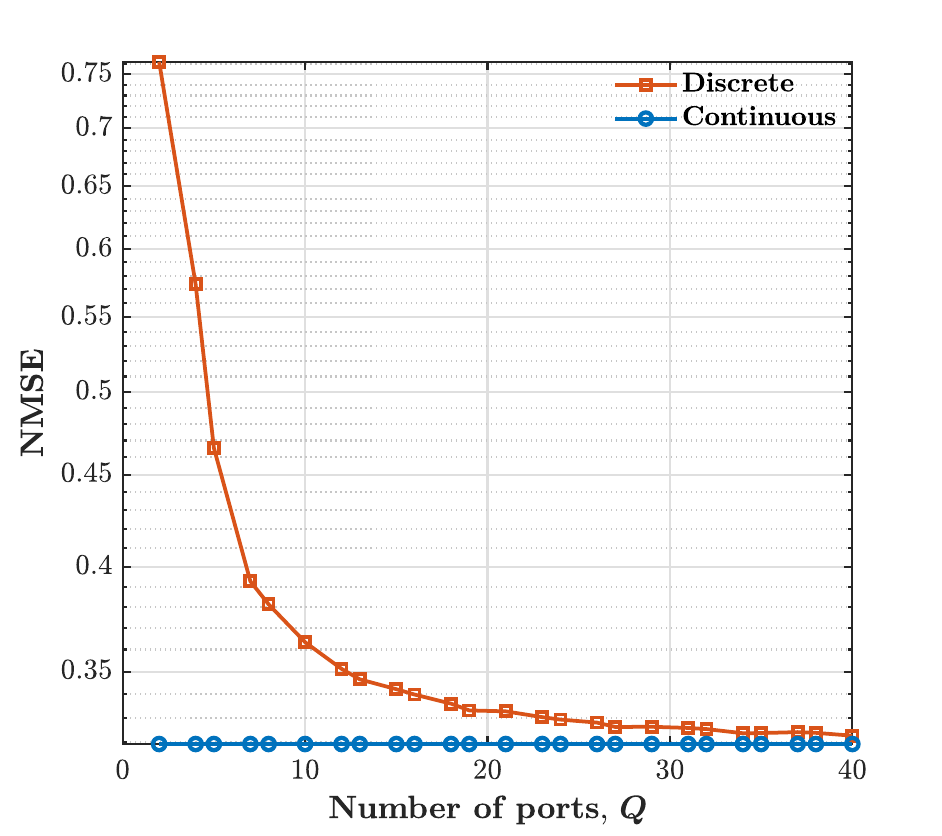}
			\label{fig:n-p}
		}\hfil\hspace{-5mm}
		\caption{(a) CDF versus NMSE; (b) NMSE versus pilot length $\tau_p$; (c) NMSE versus FA length $\ell$; and (d) NMSE versus number of ports $Q$.}
		\label{fig:main}\vspace{0mm}
	\end{figure*}
	\section{Comparison with Discrete Port-Based FA}
	Here, we compare the proposed continuous FA sampling framework with conventional discrete port-based FA architectures commonly considered in the literature.
	\subsection{Discrete Port-Based FA Model}
	
	In a discrete FA architecture, the antenna position at each AP is restricted to a finite set of $Q$ predefined port locations
	$
	\mathcal{X}_Q = \{x_1, x_2, \dots, x_Q\} \subset [0,\ell].
	$
	During uplink training, the FA can only occupy positions in $\mathcal{X}_Q$, and the sampling locations $\{x_l(u)\}_{u=1}^{\tau_p}$ must satisfy
	\begin{align}
	x_l(u) \in \mathcal{X}_Q, \quad u = 1,\dots,\tau_p.
	\end{align}
	Moreover, to ensure a fair comparison with the continuous FA model in \eqref{eq:motion_constraint}, we impose the same per-symbol position constraint on discrete switching, i.e.,
	\begin{equation}
		|x_l(u+1)-x_l(u)| \le v_{\max}T_s,\quad u=1,\dots,\tau_p-1,
		\label{eq:motion_constraint_discrete}
	\end{equation}
	with $x_l(u)\in\mathcal{X}_Q$. This captures practical limitations of port switching and ensures that the feasible discrete sampling trajectories form a subset of the continuous admissible trajectories. Typical discrete FA schemes select ports either deterministically, e.g., round-robin or adaptively based on heuristic criteria such as instantaneous or average channel quality. Under this restriction, the channel samples collected during training correspond to noisy observations of the underlying spatial random field at a finite set of locations. As a result, channel estimation is effectively performed using a finite-dimensional approximation of the continuous spatial process.
	
	\subsection{Continuous FA Sampling Model}
	
	In contrast, the proposed continuous FA framework allows the antenna position $x_l(u)$ to take any value in the continuous domain $[0,\ell]$, subject only to the physical position constraint in \eqref{eq:motion_constraint}. This enables arbitrary spatial sampling of the channel during pilot transmission and eliminates the need for discrete port switching. From a statistical perspective, continuous FA sampling allows the observation locations to be selected from an uncountable set, thereby enlarging the class of admissible sampling strategies relative to discrete FA schemes.
	
	\subsection{Fundamental Performance Comparison}
	
	
	\begin{proposition}
		\label{pro-1}
		Fix $\tau_p$ and assume that both continuous and discrete FA sampling satisfy the same position constraint, i.e., \eqref{eq:motion_constraint} and \eqref{eq:motion_constraint_discrete}, respectively. Then, the minimum achievable NMSE under continuous FA sampling is no larger than that of any discrete port-based scheme with a finite number of ports $Q$. Moreover, for any non-degenerate spatial correlation kernel, the inequality is strict unless the discrete port set is sufficiently dense (as $Q\to\infty$).
	\end{proposition}
	
%
	
	\begin{proof}
		Under \eqref{eq:motion_constraint_discrete}, any discrete sampling trajectory $\{x_l(u)\}_{u=1}^{\tau_p}$ is admissible for the continuous model, since $\mathcal{X}_Q\subset[0,\ell]$ and both satisfy the same position constraint. Hence, the set of feasible discrete trajectories is a subset of the feasible continuous trajectories. Since the NMSE in \eqref{eq:nmse_expression} is the LMMSE error for a Gaussian process and depends on the sampling locations only through the resulting covariance matrices, minimizing over a superset of feasible trajectories cannot yield a larger optimum. Therefore, the minimum NMSE achievable by continuous sampling is no larger than that achievable by any discrete scheme with finite $Q$.
		
		For a non-degenerate kernel, restricting sampling locations to a finite set $\mathcal{X}_Q$ yields a strictly smaller attainable observation subspace than allowing arbitrary locations in $[0,\ell]$, which leads to strictly larger posterior variance unless $\mathcal{X}_Q$ becomes dense as $Q\to\infty$.
	\end{proof}


\begin{corollary}
	Under the same position constraint, discrete port-based schemes approach the continuous-sampling limit only as $Q$ increases and the port set becomes dense in $[0,\ell]$.
\end{corollary}

	
	\section{Trajectory Design and Complexity Discussion}
	
	\subsection{Trajectory Design Perspective}
	
	Under the continuous FA framework, the uplink training problem can be interpreted as a spatial sampling design problem for a Gaussian random field. The FA trajectory
$
	\{x_l(u)\}_{u=1}^{\tau_p}
$
	determines the sampling locations and directly influences the posterior variance in \eqref{eq:nmse_expression}. Optimal trajectory design thus aims to select sampling locations that maximize the information content of the observations subject to the position constraint in \eqref{eq:motion_constraint}. While jointly optimal trajectory design is in general a non-convex problem, simple structured trajectories are sufficient to capture most of the available spatial diversity. In particular, linear sweep trajectories and periodic oscillatory motions uniformly cover the admissible spatial domain and achieve performance close to the best observed among simple admissible trajectories in practice. These trajectories exploit the spatial correlation structure without requiring instantaneous position optimization or feedback.\vspace{-4mm}
	\subsection{Computational Complexity and Practical Considerations}
	The dominant computational cost of LMMSE channel estimation arises from the inversion of the $\tau_p \times \tau_p$ covariance matrix in \eqref{eq:lmmse_estimator}, which scales as $\mathcal{O}(\tau_p^3)$ per UE-AP pair. This complexity is identical to that of discrete port-based FA schemes and independent of the spatial resolution of the FA. Importantly, continuous FA sampling eliminates the need for fast port switching and associated control signaling. As a result, the proposed framework achieves superior estimation performance without increasing computational complexity or requiring additional RF chains, making it well suited for practical CF-mMIMO deployments.
It is noteworthy that the proposed framework naturally accommodates realistic position constraints through \eqref{eq:motion_constraint}. As long as the FA trajectory covers a sufficiently diverse set of spatial locations within the pilot duration, the performance gains of continuous FA sampling persist even under stringent velocity limitations.
	\section{Numerical Results}
Here, we evaluate the numerical results for the proposed CF-mMIMO system with $L=64$ access points and $K=10$ users uniformly distributed over a $400\times 400$~m$^2$ area. Large-scale fading includes distance-dependent path loss with exponent $\alpha=3.2$ and log-normal shadowing with standard deviation $8$~dB. Unless otherwise stated, the channel estimation error is evaluated at a target antenna position $x^\star=\ell/2$, corresponding to the position used for uplink data transmission. The default simulation parameters are $\sigma^2=1$, $\ell=2$, $\tau_p=10$, $\eta_p=10$ (pilot SNR $=10$~dB), and $Q=8$ uniformly spaced ports. The FA position constraint is defined in \eqref{eq:motion_constraint} with $T_s=1$ and $v_{\max}=0.3$, and results are averaged over multiple independent network realizations.

	Fig.~\ref{fig:main}(a) illustrates the CDF of the NMSE achieved by the proposed continuous FA sampling scheme and by discrete port-based architectures with different numbers of antenna ports $Q$. As shown in Fig.~1(a), continuous FA sampling consistently yields a left-shifted CDF relative to all discrete schemes, indicating uniformly lower estimation error across the entire distribution. Increasing the number of discrete ports improves performance and gradually narrows the gap; however, see that even with a relatively large number of ports, discrete schemes remain inferior to continuous sampling. This behavior is a direct consequence of the fundamentally different spatial sampling capabilities of the two architectures. Continuous FA motion enables channel observations to be collected at arbitrary spatial locations, allowing the LMMSE estimator to exploit a richer observation subspace of the underlying spatially correlated channel. In contrast, discrete port-based schemes restrict sampling to a finite set of locations, which inherently limits the dimensionality of the observation space and results in information loss. As a result, discrete architectures can only approach the performance of continuous sampling as the number of ports becomes large, at the cost of increased hardware complexity and switching overhead.
	
	Fig.~\ref{fig:main}(b) shows the NMSE as a function of the pilot length $\tau_p$ for continuous and discrete fluid-antenna architectures. As expected, the estimation error decreases monotonically with $\tau_p$ for both schemes due to increased observation diversity. However, continuous FA sampling consistently achieves lower NMSE across all pilot lengths. This gap widens as $\tau_p$ increases, since additional pilots allow continuous sampling to explore a larger set of spatial locations, whereas discrete schemes remain constrained to a finite port set. Consequently, the marginal benefit of longer training is fundamentally higher for continuous FA sampling.
	
	Fig.~\ref{fig:main}(c) shows the NMSE as a function of the fluid-antenna length $\ell$ for a fixed pilot length $\tau_p$. As $\ell$ increases, the spatial separation between successive sampling locations grows, leading to weaker spatial correlation among the observed channel samples. Since the LMMSE estimator relies on spatial correlation to infer the channel at a given position, this reduced correlation results in a gradual increase in estimation error for both continuous and discrete schemes. Nevertheless, continuous FA sampling consistently achieves lower NMSE across all antenna lengths, as it exploits the available spatial domain more efficiently than discrete port-based sampling. The widening performance gap reflects the increasing spatial undersampling penalty suffered by discrete architectures as $\ell$ grows.
	
	Fig.~\ref{fig:main}(d) illustrates the NMSE as a function of the number of discrete antenna ports $Q$. As $Q$ increases, the performance of the discrete architecture improves and gradually approaches that of continuous FA sampling. This behavior reflects the fact that a denser port set provides a finer approximation of continuous spatial sampling. Nevertheless, for any finite $Q$, continuous FA sampling achieves lower NMSE, confirming that discrete port-based architectures converge to the continuous limit only as $Q$ becomes large (see Proposition \ref{pro-1}).
	\vspace{0mm}

%
	\section{Conclusion}
	We developed a continuous FA framework for uplink channel estimation in CF-mMIMO systems. By interpreting FA motion as spatial sampling of a spatially correlated random field, we showed that continuous sampling achieves equal or lower estimation error than discrete port-based architectures for any finite pilot budget under identical position constraints, with strict gains under mild and practically relevant conditions. The proposed framework provides a principled foundation for trajectory-aware antenna reconfigurability and motivates future work on joint motion and training design in distributed MIMO systems.
	\vspace{0mm}
\bibliographystyle{IEEEtran}

\end{document}